\documentclass{article}

\PassOptionsToPackage{numbers, compress, sort}{natbib}

\usepackage[final]{nips_2018}

\usepackage[utf8]{inputenc} 
\usepackage[T1]{fontenc}    
\usepackage[colorlinks=true,breaklinks=true,bookmarks=true,urlcolor=blue,
     citecolor=blue,linkcolor=blue,bookmarksopen=false,draft=false]{hyperref}
\usepackage{url}            
\usepackage{booktabs}       
\usepackage{amsfonts}       
\usepackage{nicefrac}       
\usepackage{microtype}      
\usepackage{tikz,pgfplots}
\usetikzlibrary{arrows,positioning,calc}

\usepackage[colorinlistoftodos,textsize=tiny]{todonotes}
\usepackage{color}
\usepackage{amsmath,amsthm}

\newcommand{\Comments}{0}
\definecolor{gray}{gray}{0.5}
\definecolor{darkgreen}{rgb}{0,0.5,0}
\newcommand{\mynote}[2]{\ifnum\Comments=1\textcolor{#1}{#2}\fi}

\newcommand{\mytodo}[2]{\ifnum\Comments=1%
  \fi}
\newcommand{\raft}[1]{\mytodo{darkgreen!20!white}{RF: #1}}
\renewcommand{\bot}[1]{\mytodo{red!20!white}{BW: #1}}

\newtheorem{theorem}{Theorem}
\newtheorem{lemma}{Lemma}
\newtheorem{claim}{Claim}

\theoremstyle{definition}
\newtheorem{definition}{Definition}

\newcommand{\F}{\mathcal{F}}
\newcommand{\X}{\mathcal{X}}
\newcommand{\Y}{\mathcal{Y}}
\newcommand{\Z}{\mathcal{Z}}

\DeclareMathOperator*{\E}{\mathbb{E}}
\DeclareMathOperator*{\argmax}{\text{argmax}}

\newcommand{\ones}{\mathbf{1}} 

\newcommand{\reals}{\mathbb{R}}

\newcommand{\z}{z}

\def\part{\textsc{Participant}}

\renewcommand{\Comments}{0}

\title{Bounded-Loss Private Prediction Markets}

\author{
  Rafael Frongillo \\
  Colorado Boulder\\
  \texttt{raf@colorado.edu} \\
  \And
  Bo Waggoner \\
  Microsoft Research\\
  \texttt{bwag@colorado.edu}
}

\begin{document}
\maketitle

\begin{abstract}
  Prior work has investigated variations of prediction markets that preserve participants' (differential) privacy, which formed the basis of useful mechanisms for purchasing data for machine learning objectives.
  Such markets required potentially unlimited financial subsidy, however, making them impractical.
  In this work, we design an adaptively-growing prediction market with a bounded financial subsidy, while achieving privacy, incentives to produce accurate predictions, and precision in the sense that market prices are not heavily impacted by the added privacy-preserving noise.
  We briefly discuss how our mechanism can extend to the data-purchasing setting, and its relationship to traditional learning algorithms.
\end{abstract}

\section{Introduction}

In a prediction market, a platform maintains a prediction (usually a probability distribution or an expectation) of a future random variable such as an election outcome.
Participants' trades of financial securities tied to this event are translated into updates to the prediction.
Prediction markets, designed to aggregate information from participants, have gained a substantial following in the machine learning literature.
One reason is the overlap in goals (predicting future outcomes) as well as techniques (convex analysis, Bregman divergences), even at a deep level: the form of market updates in standard automated market makers have been shown to mimic standard online learning or optimization algorithms in many settings~\cite{frongillo2012interpreting,abernethy2013efficient,abernethy2014information,frongillo2015convergence}.
Beyond this research-level bridge, recent papers have suggested prediction market mechanisms as a way of \emph{crowdsourcing} data or algorithms for machine learning, usually by providing incentives for participants to repeatedly update a centralized hypothesis or prediction~\cite{abernethy2011collaborative,waggoner2015market}.

One recently-proposed mechanism to purchase data or hypotheses from participants is that of Waggoner, et al.~\cite{waggoner2015market}, in which participants submit updates to a centralized market maker, either by directly altering the hypothesis, or in the form of submitted data; both are interpreted as buying or selling shares in a market, paying off according to a set of holdout data that is revealed after the close of the market.
The authors then show how to preserve \emph{differential privacy} for participants, meaning that the content of any individual update is protected, as well as natural accuracy and incentive guarantees.

One important drawback of Waggoner, et al.~\cite{waggoner2015market}, however, is the lack of a \emph{bounded worst-case loss} guarantee:
as the number of participants grows, the possible financial liability of the mechanism grows without bound.
In fact, their mechanism cannot achieve a bounded worst-case loss without giving up privacy guarantees.
Subsequently, Cummings, et al.~\cite{cummings2016possibilities} show that \emph{all} differentially-private prediction markets of the form proposed in~\cite{waggoner2015market} must suffer from unbounded financial loss in the worst case.
Intuitively, one could interpret this negative result as saying that the randomness of the mechanism, which must be introduced to preserve privacy, also creates \emph{arbitrage} opportunities for participants: by betting against the noise, they collectively expect to make an unbounded profit from the market maker.
Nevertheless, Cummings, et al.\ leave open the possibility that mechanisms outside the mold of Waggoner, et al.\ could achieve both privacy and a bounded worst-case loss.

In this paper, we give such a mechanism: the first private prediction market framework with a bounded worst-case loss.
When applied to the crowdsourcing problems stated above, this now allows the mechanism designer to maintain a fixed budget.
Our construction and proof proceeds in two steps.

We first show that by adding a small transaction fee to the mechanism of \cite{waggoner2015market}, one can eliminate financial loss due to arbitrage while maintaining the other desirable properties of the market.
The key idea is that a carefully-chosen transaction fee can make each trader subsidize (in expectation) any arbitrage that may result from the noise preserving her privacy.
Unless prices already match her beliefs quite closely, however, she still expects to make a profit by paying the fee and participating.
We view this as a positive result both conceptually---it shows that arbitrage opportunities are not an insurmountable obstacle to private markets---and technically---the designer budget grows very slowly, only $O((\log T)^2)$, with the number of participants $T$.

Nonetheless, this first mechanism is still not completely satisfactory, as the budget is superconstant in $T$, and $T$ must be known in advance.
This difficulty arises not from arbitrage, but (apparently) a deeper constraint imposed by privacy that forces the market to be large relative to the participants.
\raft{What do you mean by ``large''?  Depth, liquidity, small sensitivity?  (I assume the latter, but wanted to be sure; in any case we can say ``depth'' in quotes.)}\bot{I guess small sensitivity, but kind of all of them...}
Our second and main result overcomes this final hurdle.
We construct a sequence of adaptively-growing markets that are syntactically similar to the ``doubling trick'' in online learning.
The key idea is that, in the market from our first result, only about $(\log T)^2$ of the $T$ participants can be \emph{informational} traders; after this point, additional participants do not cost the designer any more budget, yet their transaction fees can raise significant funds.
So if the end of a stage is reached, the market activity has actually generated a surplus which subsidizes the initial portion of the next stage of the market.

\section{Setting}

In a cost-function based prediction market, there is an observable future outcome $Z$ taking values in a set $\Z$.
The goal is to predict the expectation of a random variable $\phi: \Z \to \reals^d$.
We assume $\phi$ is a bounded random variable, as otherwise prediction markets with bounded financial loss are not possible.
Participants will buy from the market \emph{contracts}, each parameterized by a vector $r \in \reals^d$.
The contract represents a promise for the market to pay the owner $r \cdot \phi(Z)$ when $Z$ is observed.
Adopting standard financial terminology, in our model there are $d$ \emph{securities} $j=1,\dots,d$, and the owner of a \emph{share} in security $j$ will receive a payoff of $\phi(Z)_j$, that is, the $j$th component of the random variable.
Thus a contract $r \in \reals^d$ contains $r_j$ shares of security $j$ and pays off $\sum_{j=1}^d r_j \phi(Z)_j = r \cdot \phi(Z)$.
Note that $r_j < 0$, or ``short selling'' security $j$, is allowed.

The market maintains a \emph{market state} $q^t \in \reals^d$ at time $t=0,\dots,T$, with $q^0 = 0$.
Each trader $t=1,\dots,T$ arrives sequentially and purchases a contract $dq^t \in \reals^d$, and the market state is updated to $q^t = q^{t-1} + dq^t$.
In other words, $q^t = \sum_{s=1}^t dq^s$, the sum of all contracts purchased up to time $t$.
The price paid by each participant is determined by a convex \emph{cost function} $C: \reals^d \to \reals$.
Intuitively, $C$ maps $q^t$ to the total price paid by all agents so far, $C(q^t)$.
Thus, participant $t$ making trade $dq^t$ when the current state is $q^{t-1}$ pays $C(q^{t-1} + dq^t) - C(q^{t-1})$.
Notice that the \emph{instantaneous prices} $p^t = \nabla C(q^t)$ represent the current price per unit of infinitesimal purchases, with the $j$th coordinate representing the current price per share of the $j$th security.

The prices $\nabla C(q)$ are interpreted as predictions of $\E \phi(Z)$, as an agent who believes the $j$th coordinate is too low will purchase shares in it, raising its price, and so on.
This can be formalized through a learning lens: It is known~\citep{abernethy2013efficient} that agents in such a market maximize expected profit by minimizing an expected Bregman divergence between $\phi(Z)$ and $\nabla C(q)$; of course, it is known that $\nabla C(q) = \E \phi(Z)$ minimizes risk for any divergence-based loss~\cite{savage1971elicitation,banerjee2005optimality,abernethy2012characterization}.
(The Bregman divergence is that corresponding to $C^*$, the convex conjugate of $C$.)

\textbf{Price Sensitivity.}~
The price sensitivity of a cost function $C$ is a measure of how quickly prices respond to trades, similar to ``liquidity'' discussed in \citet{abernethy2013efficient,abernethy2014general} and earlier works.
Formally, the \emph{price sensitivity} $\lambda$ of $C$ is the supremum of the operator norm of the Hessian of $C$, with respect to the $\ell_1$ norm.\footnote{For convenience we will assume $C$ is twice differentiable, though this is not necessary.}
In other words, if $c = \|q-q'\|_1$ shares are purchased, then the change in prices $\| \nabla C(q) - \nabla C(q') \|_1$ is at most $\lambda c$.

Price sensitivity is directly related to the worst-case loss guarantee of the market, as follows.
Those familiar with market scoring rules may recall that with scoring rule $S$, the loss can be bounded by (a constant times) the largest possible score.
Hence, scaling $S$ by a factor $\frac{1}{\lambda}$ immediately scales the loss bound by $\frac{1}{\lambda}$ as well.
Recall that $S$ is defined by a convex function $G$, the convex conjugate of $C$.
Scaling $S$ by $\frac{1}{\lambda}$ is equivalent to scaling $G$ by $\frac{1}{\lambda}$.
By standard results in convex analysis, this is equivalent to transforming $C$ into $C_\lambda(q) = \frac{1}{\lambda} C\left(\lambda q \right)$, an operation known as the perspective transform.  This in turn scales the price sensitivity by $\lambda$ by the properties of the Hessian.

Price sensitivity is also related to the total number of trades required to change the prices in a market.
If we assume each trade consists of at most one share in each security, then $\frac{1}{\lambda}$ trades are necessary to shift the predictions to an arbitrary point from an arbitrary point.

\textbf{Convention: normalized, scaled $C$.}~
In the remainder of the paper, we will suppose that we start with some convex cost function $C_1$ whose price sensitivity equals $1$ and worst-case loss bounded by some constant $B_1$.
Then, to obtain price sensitivity $\lambda$, we use the cost function $C(\cdot) = \frac 1 \lambda C_1(\lambda \cdot)$.
As discussed above, $C$ has price sensitivity at most $\lambda$ and a worst-case loss bound of $B = B_1/\lambda$.
(This assumption is without loss of generality, as any cost function that guarantees a bounded worst-case loss can be scaled such that its price sensitivity is $1$.)

\subsection{Prior work}
To achieve differential privacy for trades of a bounded size (which will be assumed), the general approach is to add random noise to the ``true'' market state $q$ and publish this noisy state $\hat{q}$.
The privacy level thus determines how close $\hat{q}$ is to $q$.
The distance from $\nabla C(\hat{q})$ to $\nabla C(q)$ is then controlled by the price sensitivity $\lambda$.
For a fixed noise and privacy level, a smaller $\lambda$ leads to small impact of noise on prices, meaning very good accuracy.
However, decreasing $\lambda$ does not come for free: the worst-case financial loss of to the market designer scales as $1/\lambda$.

The market of~\cite{waggoner2015market} adds controlled and correlated noise over time, in a manner similar to the ``continual observation'' technique of differential privacy.
This reduces the influence of noise on accuracy to polylogarithmic in $T$, the number of participants.
Their main result for the prediction market setting studied here is as follows.
\begin{theorem}[\cite{waggoner2015market}] \label{thm:classic-summary}
  Assuming that all trades satisfy $\| dq^t \|_1 \leq 1$, the private mechanism is $\epsilon$-differentially private in the trades $dq^1,\dots,dq^T$ with respect to the output $\hat{q}^1,\dots,\hat{q}^T$.
  Further, to satisfy $\|p^t - \hat{p}^t\|_1 \leq \alpha$ for all $t$, except with probability $\gamma$, it suffices for the price sensitivity to be
  \begin{align}\label{eq:lambda-star}
    \lambda^* = \frac{\alpha ~ \epsilon}{4\sqrt{2} d \lceil \log T \rceil \ln(2Td/\gamma)} ~.
  \end{align}
\end{theorem}

\subsection{Our setting and desiderata}
This paper builds on the work of \citet{waggoner2015market} to overcome the negative results of \citet{cummings2016possibilities}.
Here, we formalize our setting and four desirable properties we hope to achieve.

Write a prediction market mechanism as a function $M$ taking inputs $\vec{dq}=dq^1,\dots,dq^T$ and outputting a sequence of market states $\hat{q}^1,\dots,\hat{q}^T$.
Here $\hat{q}^t$ is thought of as a noisy version of $q^t = \sum_{s\leq t} dq^s$.
Each of these states is associated with a prediction $\hat{p}^t$ in the set of possible prices (expectations of $\phi$), while the state $q^t$ is associated with the ``true'' underlying prediction $p^t$.
\begin{definition}[Privacy]
  $M$ satisfies \emph{$(\epsilon,\delta)$-differential privacy} if for all pairs of inputs $\vec{dq},\vec{dq'}$ differing by only a single participants' entry, and for all sets $S$ of possible outputs, $\Pr[M(\vec{dq}) \in S] \leq e^{\epsilon} \Pr[M(\vec{dq'}) \in S] + \delta$.
  If furthermore $\delta = 0$, we say $M$ is \emph{$\epsilon$-differentially private}.
\end{definition}
\begin{definition}[Precision]
  $M$ has \emph{$(\alpha,\gamma)$ precision} if for all $\vec{dq}$, with probability $1-\gamma$, we have $\|\hat{p}^t - p^t\|_1 \leq \alpha$ for all $t$.
\end{definition}
\begin{definition}[Incentives]
  $M$ has \emph{$\beta$-incentive} to participate if, for all beliefs $p = \E \phi(Z)$, if at any point $\|\hat{p}^t - p \|_{\infty} > \beta$, then there exists a participation opportunity that makes a strictly positive profit in expectation with respect to $p$.
\end{definition}
For the budget guarantee, we must formalize the notion that participants may respond to the noise introduced by the mechanism.
Following \citet{cummings2016possibilities}, let a \emph{trader strategy} $\vec{s} = (s^1,\dots,s^T)$ where each $s^t$ is a possibly-randomized function of the form $s^t(dq^1,\dots,dq^{t-1};\hat{q}^1,\dots,\hat{q}^{t-1}) = dq^t$, i.e. a strategy taking the entire history prior to $t$ and outputting a trade $dq^t$.
Let $L(M, \vec{s}, z)$ be a random variable denoting the financial loss of the market $M$ against trader strategy $\vec{s}$ when $Z=z$, which for the mechanism described above is simply
\centerline{$\displaystyle L(M, \vec{s}, z) = \sum_{t=1}^T \left[ C(\hat{q}^t) - C(\hat{q}^t + dq^t) - dq^t \cdot \phi(z) \right]$~.}
\begin{definition}
  $M$ guarantees \emph{designer budget $B$} if, for any trader strategy $\vec{s}$ and all $z$, $\E L(M, \vec{s}, z) \leq B$, where the expectation is over the randomness in $M$ and each $s^t$.
\end{definition}

\section{Slowly-Growing Budget} \label{sec:trans-fee}
The private market of \citet{waggoner2015market} causes unbounded loss for the market maker in two ways.
The first is from traders betting against the random noise introduced to protect privacy.
This is a key idea leveraged by \citet{cummings2016possibilities} to show negative results for private markets.
In this section, we show that a transaction fee can be chosen to exactly balance the expected profit from this type of arbitrage.\footnote{Intuitively, it is enough for the fee to cover arbitrage amounts in expectation, because a trader must pay the fee to trade before the random noise is drawn and any arbitrage opportunity is revealed.}
We will show that this fee is still small enough to allow for very accurate prices.\footnote{For instance, if the current price of a security is $0.49$ and a trader believes the true price should be $0.50$, she will purchase a share if the fee is $c < 0.01$.
    (For privacy, we limit each trade to a fixed size, say, one share.)} 
This transaction fee restores the worst-case loss guarantee to the inverse of the price sensitivity, just as in a non-private market.
The second way the market causes unbounded loss is to require price sensitivity to shrink as a function of $T$; this is addressed in the next section.

We show that with this carefully-chosen fee, the market still achieves precision, incentive, and privacy guarantees, but now with a worst-case market maker loss of $O((\log T)^2)$, much improved over the na\"ive $O(T)$ bound.
This is viewed as a positive result because the worst-case loss is growing quite slowly in the total number of participants, and moreover matches the fundamental ``informational'' worst-case loss one expects with price sensitivity $\lambda^*$.

\subsection{Mechanism and result}
Here we recall the private market mechanism of \cite{waggoner2015market}, adapted to the prediction market setting following \cite{cummings2016possibilities}.
We will express the randomness of the mechanism in terms of a ``noise trader'' for both intuition and technical convenience.
The market is defined by a cost function $C$ with price sensitivity $\lambda$, and parameters $c$ (transaction fee), $\epsilon$ (privacy), $\alpha, \gamma$ (precision), and $T$ (maximum number of participants).
There is a special trader we call the \emph{noise trader} who is controlled by the designer.
All actions of the noise trader are hidden and known only by the designer.
The designer publishes an initial market state $q^0 = \hat{q}^0 = 0$.
Let $T'$ denote the actual number of arrivals, with $T' \leq T$ by assumption.
Then, for $t=1,\dots,T'$:
\begin{enumerate}\setlength{\itemsep}{0pt}
  \item Participant $t$ arrives, pays a fee of $c$, and purchases bundle $dq^t$ with $\|dq^t\|_1 \leq 1$.
        The payment is $C(\hat{q}^t + dq^t) - C(\hat{q}^t)$.
  \item The noise trader purchases a randomly-chosen bundle $z^t$, called a noise trade, after selling off some subset $\{z^{t_1},\ldots,z^{t_k}\}$ of previously purchased noise trades for $t_i < t$, according to a predetermined schedule described below.
        Letting $w^t = z^t - \sum_{i=1}^k z^{t_i}$ denote this net noise bundle, the noise trader is thus charged $C(\hat{q}^t + dq^t + w^t) - C(\hat{q}^t + dq^t)$.
  \item The ``true'' market state is updated to $q^t = q^{t-1} + dq^t$, but is not revealed.
  \item The noisy market state is updated to $\hat{q}^t = \hat{q}^{t-1} + dq^t + w^t$ and is published.
\end{enumerate}
Finally, $z \in \Z$ is observed and each participant $t$ receives a payment $dq^t \cdot \phi(z)$.
For the sake of budget analysis, we suppose that at the close of the market, the noise trader sells back all of her remaining bundles; letting $w^{T'}$ be the sum of these bundles, she is charged $C(\hat{q}^{T'} - w^{T'}) - C(\hat{q}^{T'})$.

\textbf{Noise trades.}~
Each $z^t$ is a $d$-dimensional vector with each coordinate drawn from an independent Laplace distribution with parameter $b = 2\lceil \log T \rceil / \epsilon$.
To determine which bundles $z^s$ are sold at time $t$, write $t = 2^j m$ where $m$ is odd, and sell all bundles $z^s$ purchased during the previous $2^{j-1}$ time steps which are not yet sold.
Thus, the noise trader will sell bundles purchased at times $s = t-1, t-2, t-4, t-8, \dots, t-2^{j-1}$; in particular, when $t$ is odd we have $j=0$, so no previous bundles will be sold.

\textbf{Budget.}~
The total loss of the market designer can now be written as the sum of three terms: the loss of the market maker, the loss of the noise trader, and the gain from transaction fees.
By convention, the noise trader eventually sells back all bundles it purchases and is left with no shares remaining.
\begin{align}\label{eq:designer-budget}
  \!\!\!\!\!\!
  L(M, \vec{s}, z) = \overbrace{\sum_{t=1}^{T'} C(\hat{q}^{t-1}) - C(\hat{q}^{t-1} \!\!+ dq^t) + dq^t \!\cdot \phi(z)}^{\text{net loss of market maker}} + \overbrace{\sum_{t=1}^{T'} C(\hat{q}^{t-1}\!\!+dq^t) - C(\hat{q}^t)}^{\text{net loss of noise trader}} + \!\overbrace{c T'\vphantom{\sum_{t=1}^{T'}}}^{\text{fees}}\!\!.\!\!
\end{align}
The main result of this section is as follows.
\begin{theorem} \label{thm:private-loss-one-over-lambda}
  When each arriving participant pays a transaction fee $c = \alpha$, the private market with any $\lambda \leq \lambda^*$ from eq.~\eqref{eq:lambda-star} satisfies $\epsilon$-differential privacy, $(\alpha,\gamma)$-precision, $2\alpha$-incentive to trade, and budget bound $\frac{B_1}{\lambda}$, where $B_1$ is the budget bound of the underlying cost function $C_1$.
\end{theorem}

\subsection{Proof ideas: privacy, precision, incentives}
The differential privacy and precision claims follow directly from the prior results, as nothing has changed to impact them.
The incentive claim is not technically involved, but perhaps subtle: the transaction fee should be high enough to eliminate expected profit from arbitrage, yet low enough to allow for profit from information.
The key point is that the transaction fee is a constant, but the farther the prices are from the trader's belief, the more money she expects to make from a constant-sized trade.
The transaction fee creates a ball of size $2\alpha$ around the current prices where, if one's belief lies in that ball, then participation is not profitable.

We give most of the proof of the designer budget bound, with some claims deferred to the full version.
\begin{lemma}[Budget bound] \label{lemma:fee-wcl}
  The transaction-fee private market with any price sensitivity $\lambda \leq \lambda^*$ guarantees a designer budget bound of $\frac{B_1}{\lambda}$.
\end{lemma}
\begin{proof}
  Let $c$ be the transaction fee; we will later take $c=\alpha$.
  Then the worst-case loss from eq.~\eqref{eq:designer-budget} is
  \[ WC(\lambda,T') := WC_0(\lambda, T') + NTL(\lambda,T') - T'c~, \]
  where $WC_0(\lambda, T')$ is the worst-case loss of a standard prediction market maker with parameter $\lambda$ and $T'$ participants, $NTL(\lambda,T')$ is the worst-case noise trader loss, and $T'c$ is the revenue from $T'$ transaction fees of size $c$ each.

  The worst-case loss of a standard prediction market maker is well-known; see e.g.~\cite{abernethy2013efficient}.
  By our normalization and definition of price sensitivity, we thus have $WC_0(\lambda, T') \leq \frac{B_1}{\lambda}$.

  To bound the noise trader loss $NTL(\lambda, T')$, we will consider each bundle $z^t$ purchased by the noise trader.
  The idea is to bound the difference in price between the purchase and sale of $z^t$.
  For analysis, we suppose that at each $t$, the noise trader first sells any previous bundles (e.g.\ at $t=4$, first selling $z^3$ and then selling $z^2$), and then purchases $z^t$.

  Now let $b(t)$ be the largest power of $2$ that divides $t$.
  Let $q_{\text{buy}}^t$ and $q_{\text{sell}}^t$ be the market state just before the noise trader purchases $z^t$ and just after she sells $z^t$, respectively.
  \begin{claim} \label{claim:traders-between-noise-bundle}
    For each $t$, exactly $b(t)$ traders arrive between the purchase and the sale of bundle $z^t$; furthermore, $q_{\text{sell}}^t - q_{\text{buy}}^t$ is exactly equal to the sum of these participants' trades.
  \end{claim}
  For example, suppose $t$ is odd.
  Then only one participant arrives between the purchase and sale of $z^t$.
  Furthermore, $z^t$ is the last bundle purchased by the noise trader at time $t$ and is the first sold at time $t+1$, so the difference in market state is exactly $z^t$ plus that participant's trade.
  
  \begin{claim} \label{claim:max-gap}
    If the noise trader purchases and later sells $z^t$, then her net loss in expectation over $z^t$ (but for any trader behavior in response to $z^t$), is at most $\lambda b(t) K$ where $K = \E \|z^t\|_2$.
  \end{claim}
    We now sum over all bundles $z^t$ purchased by the noise trader, i.e. at time steps $1,\dots,T'$.
  Recall that the noise trader sells back every bundle $z^t$ she purchases.
  Thus, her total payoff is the sum over $t$ of the difference in price at which she buys $z^t$ and price at which she sells it.
  For each $j=0,\dots,\log T' - 1$, there are $2^j$ different steps $t$ with $b(t) = T'/2^{j+1}$.
  The total loss is thus,
  \begin{align}
   NTL(\lambda, T')
    &\leq \sum_{j=0}^{\log T' - 1} 2^j \frac{T'}{2^{j+1}} \lambda K
    =    \frac{T' \log T'}{2} \lambda K~. \label{eqn:noise-loss}
  \end{align}
  Note that if the noise trader has some noise bundles left over after the final participant, we suppose she immediately sells all remaining bundles back to the market maker in reverse order of purchase.

  Putting eq.~\eqref{eqn:noise-loss} together with the above bound on $WC_0$ gives
  \begin{align}
    WC(\lambda, T') &\leq WC_0(\lambda, T') + T' \log T' \lambda K - T'c 
                     \leq \frac{B_1}{\lambda} + T' \left( K \log T' \lambda - c \right)~, \label{eqn:exact-wcl-fee-bound}
  \end{align}
  which is in turn at most $B_1/\lambda$
  if we choose $\lambda$ and the transaction fee $c$ such that $c \geq K \log T \lambda$.
  In other words, we take $\lambda \leq c / K \log T$.

  Finally, we can bound $K = \E \|z^t\|_2$ from Claim~\ref{claim:max-gap} as follows: for each $t$, the components of the $d$-dimensional vector $z^t$ are each independent $\mathrm{Lap}(b)$ variables with $b = 2 \lceil \log T \rceil / \epsilon$.
  By concavity of $\sqrt{\cdot}$, we have
  \begin{align*}
    K = \E \sqrt{\sum_{i=1}^d z^t(i)^2}
      \leq \sqrt{\sum_i \E z^t(i)^2}
      = \sqrt{d \mathrm{Var}(\mathrm{Lap}(b))}
      = \sqrt{2d b^2}
      = 2 \sqrt{2d} \frac{\lceil \log T \rceil}{\epsilon}~.
  \end{align*}

  Therefore, it suffices to pick
    \[ \lambda \leq \frac{c ~ \epsilon}{2 \sqrt{2d} \lceil \log T \rceil \log T} ~. \]
  For $c = \alpha$, this is in fact accomplished by the private, accurate market choosing $\lambda \leq \lambda^*$ from eq.~{eq:lambda-star}.
\end{proof}

\textbf{Limitations of this result.}~
Unfortunately, Theorem~\ref{thm:private-loss-one-over-lambda} does not completely solve our problem: the other way that privacy impacts the market's loss is by lowering the necessary price sensitivity to $\lambda^* \approx \frac{1}{(\log T)^2}$ as mentioned above, leading to a worst-case loss growing with $T$.
It does not seem possible to address this via a larger transaction fee without giving up incentive to participate: traders participate as long as their expected profit exceeds the fee, and collectively $\Omega(1/\lambda)$ of them can arrive making consistent trades all moving the prices in the same (correct) direction, so the total payout will still be $\Omega(1/\lambda)$.

\section{Constant Budget via Adaptive Market Size}
In this section, we achieve our original goal by constructing an adaptively-growing prediction market in which each stage, if completed, subsidizes the initial portion of the next.

The market design is the following, with each $T^{(k)}$ to be chosen later.
We run the transaction-fee private market above with $T = T^{(1)}$, transaction fee $\alpha$, and price sensitivity $\lambda^{(1)} = \lambda^*(T^{(1)}, \alpha/2, \gamma/2)$ from eq.~\eqref{eq:lambda-star}.
When (and if) $T^{(1)}$ participants have arrived, we create a new market whose initial state is such that its prices match the final (noisy) prices of the previous one.
We set $T^{(2)}$ and price sensitivity $\lambda^{(2)} = \lambda^*(T^{(2)}, \alpha/4, \gamma/4)$ for the new market.
We repeat, halving $\alpha$ and $\gamma$ at each stage and increasing $T$ in a manner to be specified shortly, until no more participants arrive.

\begin{theorem} \label{thm:main}
  For any $\alpha, \gamma, \epsilon$, the adaptive market satisfies $\epsilon$-differential privacy, $2\alpha$-incentive to trade, $(\alpha,\gamma)$-accuracy, and a designer budget bound of
    \[ B \leq B_1 \frac{72\sqrt{2} d}{\alpha ~ \epsilon} \biggl(\ln \frac{4608 B_1 \sqrt{2} d^2}{\gamma \alpha^2 \epsilon} \biggr)^2 ~, \]
  where $B_1$ is the budget bound of the underlying unscaled cost function $C_1$.
\end{theorem}

\textbf{Proof idea.}~
We set $T^{(1)} = \Theta\bigl(\frac{B_1 d \ln(B_1 d / \gamma \alpha \epsilon)^2}{\alpha^2 ~ \epsilon}\bigr)$, and $T^{(k)} = 4T^{(k-1)}$ thereafter.
\raft{I feel I must inform you that semicolons typically separate two full (related) sentences; see \url{https://www.grammarly.com/blog/semicolon/} :-P}
\bot{Noted but not internalized.}
The key will be the following observation.
The total ``informational'' profit available to the traders (by correcting the initial market prices) is bounded by $O(1/\lambda)$, so if each trader expects to profit more than the transaction fee $c$, then only $O(1/\lambda c)$ traders can all arrive and simultaneously profit.
\raft{I got confused by the ``optimal prediction'' here, thinking about LMSR where the optimal prediction is unbounded.  But then I realized later it must be bounded here since the transaction fee forces the ``effective price space'' to be bounded away from the boundary of the price space, and thus the possible share vectors landing in that set is bounded.  Maybe worth mumbling something to this effect, though must less verbosely :-]}
\bot{Addressed}
Indeed, if all $T$ participants arrive, then the total profit from transaction fees is $\Theta(T)$ while the worst-case loss from the market is $O\left( (\log T)^2\right)$.

We can leverage this observation to achieve a bounded worst-case loss with an ``adaptive-liquidity'' approach, similar in spirit to~\citet{abernethy2014general} but more technically similar to the doubling trick in online learning.
Begin by setting $\lambda^{(1)}$ on the order of $1/(\log T^{(1)})^2 = \Theta(1)$, and run a private market for $T^{(1)}$ participants.
If fewer than $T^{(1)}$ participants show up, the worst-case loss is order $1/\lambda^{(1)}$, a constant.
If all $T^{(1)}$ participants arrive, then (for the right choice of constants) the market has actually turned a profit $\Omega(T^{(1)})$ from the transaction fees.
Now set up a private market for $T^{(2)} = 4T^{(1)}$ traders with $\lambda^{(2)}$ on the order of $1/(\log T^{(2)})^2$.
If fewer than $T^{(2)}$ participants arrive, the worst-case loss is order $1/\lambda^{(2)}$.
However, we will have chosen $T^{(2)}$ such that this loss is smaller than the $\Omega(T^{(1)})$ profit from the previous market.
Hence, the total worst-case loss remains bounded by a constant.

If all $T^{(2)}$ participants arrive, then again this market has turned a profit, which can be used to completely offset the worst-case loss of the next market, and so on.
Some complications arise, as to achieve $(\alpha,\gamma)$-precision, we must set $\alpha^{(1)},\gamma^{(1)},\alpha^{(2)},\gamma^{(2)},\dots$ as a convergent series summing to $\alpha$ and $\gamma$; and we must show that all of these scalings are possible in such a way that the transaction fees cover the cost of the next iteration.
(An interesting direction for future work would be to replace the iterative approach here with the continuous liquidity adaptation of~\cite{abernethy2014general}.)

More specifically, we prove that the loss in any round $k$ that is not completed (not all participants arrive) is at most $\frac{\alpha}{16}T^{(k)}$;
moreover, the profit in any round $k$ that is completed is at least $\frac{\alpha}{2}T^{(k)}$.
Of course, only one round is not completed: the final round $k$.
If $k = 1$, then the financial loss is bounded by $\frac{1}{\lambda^{(1)}}$, a constant depending only on $\alpha, \gamma, \epsilon$.
Otherwise, the total loss is the sum of the losses across rounds, but the mechanism makes a profit in every round but $k$.
Moreover, the loss in round $k$ is at most $\frac{\alpha}{2}T^{(k)}  = \frac{\alpha}{8}T^{(k-1)}$, which is at most half of the profit in round $k-1$.
So if $k \geq 2$, the mechanism actually turns a net profit.

While this result may seem paradoxical, note that the basic phenomenon appears in a classical (non-private) prediction market with a transaction fee, although to our knowledge this observation has not yet appeared in the literature.
Specifically, a classical prediction market with budget bound $B_1$, trades of size $1$, and a small transaction fee $\alpha$, will still have an $\alpha$-incentive to participate, and the worst case loss will still be $\Theta(B_1)$; this loss, however, can be extracted by as few as $\Theta(1)$ participants.
Any additional participants must be in a sense disagreeing about the correct prices; their transaction fees go toward market maker profit, but they do not contribute further to worst-case loss.

\section{Kernels, Buying Data, Online Learning} \label{sec:kernels}
While preserving privacy in prediction markets is well-motivated in the classical prediction market setting, it is arguably even more important in a setting where machine-learning hypotheses are learned from private personal data.
Waggoner et al.~\cite{waggoner2015market} develop mechanisms for such a setting based on prediction markets, and further show how to preserve differential privacy of the participants.
Yet their mechanisms are not practical in the sense that the financial loss of the mechanism could grow without bound.
In this section, we sketch how our bounded-financial-loss market can also be extended to this setting.
This yields a mechanism for purchasing data for machine learning that satisfies $\epsilon$-differential privacy, $\alpha$-precision and incentive to participate, and bounded designer budget.

To develop a mechanism which could be said to ``purchase data'' from participants, Waggoner et al.~\cite{waggoner2015market} extend the classical setting in two ways.
The first is to make the market \emph{conditional}, where we let $\Z = \X \times \Y$, and have independent markets $C_x:\reals^d\to\reals$ for each $x$.
Trades in each market take the form $q_x \in \reals^d$, which pay out $q_x\cdot \phi(y)$ upon outcome $(x',y)$ if $x=x'$, and zero if $x\neq x'$.
Importantly, upon outcome $(x,y)$, only the costs associated to trades in the $C_x$ market are tallied.

The second is to change the bidding language using a \emph{kernel}, a positive semidefinite function $k: \Z \times \Z \to \reals$.
Here we think of contracts as functions $f:\Z\to\reals$ in the reproducing kernel Hilbert space (RKHS) $\F$ given by $k$, with basis $\{f_{\z}(\cdot) = k(\z, \cdot) ~ : ~ \z \in \Z\}$.
For example, we recover the conditional market setting with independent markets with the kernel $k((x,y),(x',y')) = \ones\{x=x'\}\phi(y)\cdot\phi(y')$.
The RKHS structure is natural here because a basis contract $f_{\z}$ pays off at each $\z'$ according to the ``covariance'' structure of the kernel, i.e. the payoff of contract $f_{\z}$ when $\z'$ occurs equals $f_{\z}(\z') = k(\z,\z')$.
For example, when $\Y=\{-1,1\}$ one recovers radial basis classification using $k((x,y),(x',y')) = yy'e^{-(x-x')^2}$.

These two modifications to classical prediction markets, given as Mechanism 2 in~\cite{waggoner2015market}, have clear advantages as a mechanism to ``buy data''.
One may imagine that each agent, arriving at time $t \in \{1,\dots,T\}$, holds a data point $(x^t,y^t) \in \Z = \X \times \Y$.
A natural purchase for this agent would be a basis contract $f_{(x^t,y^t)}$, as this corresponds to a payoff that is highest when the test data point actually equals $(x^t,y^t)$ and decreases with distance as measured by the kernel structure.

The importance of privacy now becomes even more apparent, as the data point $(x^t,y^t)$ could be information sensitive to trader $t$. 
Fortunately, we can extend our main results to this setting.
To demonstrate the idea, we give a sketch of the result and proof below.

\begin{theorem}[Informal] \label{thm:kernel-conditional}
  Let $\Z = \X\times\Y$ where $\X$ is a compact subset of a finite-dimensional real vector space and $\Y$ is finite, and let positive semidefinite kernel $k:\Z\times\Z\to\reals$ be given.
  \raft{Actually $\Y$ does not need to be finite; we just need $C_x:\reals^d\to\reals$ for some $d$... maybe having a bounded $\phi:\Y\to\reals^d$ and letting $k((x,y),(x',y')) = k(x,x') \phi(y)\cdot\phi(y')$ is better.}
  For any choices of accuracy parameters $\alpha, \gamma$, privacy parameters $\epsilon, \delta$, trade size $\Delta$, and query limit $Q$, the kernel adaptive market satisfies $(\epsilon,\delta)$-differential privacy, $(\alpha,\gamma)$-precision, $2\alpha$-incentive to participate, and a bounded designer budget.
\end{theorem}
\begin{proof}[Proof Sketch]
  The precision property, i.e. that prices are approximately accurate despite privacy-preserving noise, follows from \cite[Theorem 2]{waggoner2015market}, and the technique in Theorem~\ref{thm:main} to combine the accuracy and privacy of multiple epochs.
  The incentive to trade property is essentially unchanged, as a participants' profit is still the improvement in expected Bregman divergence, which exceeds the transaction fee unless prices are already accurate.
  It thus remains only to show a bounded designer budget, which is slightly more involved.
  Briefly, Claim~\ref{claim:traders-between-noise-bundle} goes through unchanged, and Claim~\ref{claim:max-gap} holds as written where now $C$ becomes $C_x$ and $z^t$ becomes $z^t(x) = f^t(x,\cdot)$, i.e.,\ the trade at time $t$ restricted to the $C_x$ market alone.

  The remainder of Lemma~\ref{lemma:fee-wcl} now proceeds with one modification regarding the constant $K$.
  In eq.~\eqref{eqn:noise-loss}, the expression for the noise trader loss becomes $NTL(\lambda,T') = \E\bigl[\sup_{x\in\X} \sum_{t=1}^{T'} \lambda\alpha_t \|z^t(x)\|_2\bigr]$, where the $\alpha_t$ are simply coefficients to keep track of how many trades occurred between the buy and sell of noice trade $t$.
  We can proceed as follows:
  \begin{align*}
    NTL(\lambda,T')
    \leq \E\left[\sup_{x_1,\ldots,x_{T'}\in\X} \sum_{t=1}^{T'} \lambda\alpha_t \|z^t(x_t)\|_2\right]
    = \lambda\sum_{t=1}^{T'} \alpha_t \E\left[\sup_{x\in\X}\|z^t(x)\|_2\right]
    = \lambda\sum_{t=1}^{T'} \alpha_t K~,
  \end{align*}
  where $K$ is simply the constant $\E\left[\sup_{x\in\X}\|z^t(x)\|_2\right]$ where the expectation is taken over the Gaussian process generating the noise.
  \raft{There are some details here which I'm glossing over since I got too tired; long story short, I don't have a definitive reference, but I'm positive that we can restrict to reasonable classes of kernels where this will be the case.}
  It is well-known that the expected maximum of a Gaussian process is bounded~\cite{talagrand2014upper}, and thus boundedness of $K$ follows from the fact that $\Y$ is finite.
  Thus, continuing from eq.~\eqref{eqn:noise-loss} we obtain $NTL(\lambda,T') \leq \frac{T' \log T'}{2} \lambda K$ as before, with this new $K$.
  \raft{Once we have the details down, we can get the dependence on the number of securities $d$, which should be $\sqrt{d}$ but I'm not sure.}
  Finally, the proof of Theorem~\ref{thm:main} now goes through, as it only treats the mechanism from Theorem~\ref{thm:private-loss-one-over-lambda} as a black box.
\end{proof}

We close by noting the similarity between the kernel adaptive market mechanism and traditional learning algorithms, as alluded to in the introduction.
As observed by Abernethy, et al.~\cite{abernethy2013efficient}, the market price update rule for classical prediction markets resembles Follow-the-Regularized-Leader (FTRL);
specifically, the price update at time $t$ is given by $p^t = \nabla C(q^t) = \argmax_{w\in\Delta(\Y)} \langle w , \sum_{s\leq t} dq^s \rangle - R(w)$, where $dq^s$ is the trade at time $s$, and $R=C^*$ is the convex conjugate of $C$.

In our RKHS setting, we can see the same relationship.
For concreteness, let $C_x(q) = \tfrac 1 \lambda C(\lambda q)$ for all $x\in\X$, and let $R:\Delta(\Y)\to\reals$ be the conjugate of $C$.
Suppose further that each agent $t$ purchases a basis contract $df^t = f_{x^t,y^t}$, where we take a classification kernel $k'((x,y),(x',y')) = k(x,x') \ones\{y=y'\}$.
Letting $dq^t(x) = df^t(x,\cdot) \in \reals^\Y$, the market price at time $t$ is given by,
\begin{align*}
  p^t_x
  &= \argmax_{w\in\Delta(\Y)} \biggl\langle w , \sum_{s\leq t} dq^s(x) \biggr\rangle - \frac 1 \lambda R(w)
  \\
  & = \argmax_{w\in\Delta(\Y)} \biggl\langle w , \sum_{s\leq t} k((x^s,y^s),(x,\cdot)) \biggr\rangle - \frac 1 \lambda R(w)
  \\
  & = \argmax_{w\in\Delta(\Y)} \biggl\langle w , \sum_{s\leq t} k(x^s,x) \ones_{y^s} \biggr\rangle - \frac 1 \lambda R(w)~,
\end{align*}
where $\ones_y$ is an indicator vector.
Thus, the market price update follows a natural kernel-weighted FTRL algorithm, where the learning rate $\lambda$ is the price sensitivity of the market.

\section{Summary and Future Directions}
Motivated by the problem of purchasing data, we gave the first bounded-budget prediction market mechanism that achieves privacy, incentive alignment, and precision (low impact of privacy-preserving noise the predictions).
To achieve bounded budget, we first introduced and analyzed a transaction fee, achieving a slowly-growing $O((\log T)^2)$ budget bound, thus eliminating the arbitrage opportunities underlying previous impossibility results.
Then, observing that this budget still grows in the number of participants $T$, we further extended these ideas to design an adaptively-growing market, which does achieve bounded budget along with privacy, incentive, and precision guarantees.

We see several exciting directions for future work.
An extension of Theorem~\ref{thm:kernel-conditional} where $\Y$ need not be finite should be possible via a suitable generalization of Claim~\ref{claim:max-gap}.
Another important direction is to establish privacy for \emph{parameterized} settings as introduced by Waggoner, et al.~\cite{waggoner2015market}, where instead of kernels, market participants update the (finite-dimensional) parameters directly as in linear regression.
Finally, we would like a deeper understanding of the learning--market connection in nonparametric kernel settings, which could lead to practical improvements for design and deployment.

\bibliographystyle{plainnat}

\vfill
\break

\appendix

 % app
\section{Private (Unbounded-Loss) Markets} \label{sec:classic-privacy}
In this section, we review the private prediction market construction of \citet{waggoner2015market}.
We include proofs for completeness and clarity, as we focus on classic, complete cost-function based markets here whereas that paper focused on ``kernel markets'' which required additional formalism.
Our adaptive market will rely on this market construction and results.

\paragraph{Approach and notation.}
In a private market, the designer chooses an initial ``true'' market state $q^0$ (for convenience, we will assume $q^0 = 0$) and announces a ``published'' market state $\hat{q}^0 = 0$.
When participant $t=1,\dots,T$ arrives and requests trade $dq^t$, the market maker updates the true market state $q^t = q^{t-1} + dq^t$, but does not reveal the true state to anyone.
Instead, the market maker announces the published market state $\hat{q}^t$, which is some randomized function of all trades and published market states so far.
We assume that $\|dq^t\| \leq 1$ according to some norm $\|\cdot\|$, that is, each participant can buy or sell at most one ``total'' share.\footnote{One could also modify our approach to allow arbitrarily large trades, but this would also require adding proportionally large noise in order to continue to preserve privacy.}
Let $\|\cdot\|_{*}$ denote the dual norm to $\| \cdot \|$.

\paragraph{Differential privacy.}
The market mechanism can be viewed as a randomized function $M$ that takes as input a list of trades $\vec{dq} = dq^1,\dots,dq^T$ and outputs a list of published market states $\hat{q}^0,\dots,\hat{q}^T$.
We will call it \emph{$(\epsilon,\delta)$-differentially private} if changing a single participant's trade does not change the distribution on outputs much: if for all $\vec{dq}$ and $\vec{dq'}$ differing only in one entry, and for all (measurable) sets of possible outputs $S$,
  \[ \Pr\left[M\left(\vec{dq}\right) \in S\right] \leq e^{\epsilon} \Pr\left[M\left(\vec{dq}'\right) \in S\right] ~ + ~ \delta . \]
The  mechanism is $\epsilon$-differential private if it is $(\epsilon,0)$-d.p.
It is reasonable to treat $\epsilon$ as a constant whose size controls the privacy guarantee, such as $\epsilon = 0.01$.
Meanwhile, $\delta$ is normally preferred to be vanishingly small or $0$, as a mechanism can leak the private information of all individuals with $\delta$ probability and still be $(\epsilon,\delta)$-differentially private.

To be careful, we note that the market's ``full output'' also includes that it sends each participant their payoff.
However, this payoff is a function only of the public noisy market states and of that participant's trade.
The payoff is assumed to be sent privately and separately, unobservable by any other party.
By the post-processing property of differential privacy, a trader's $(\epsilon,\delta)$-privacy guarantee continues to hold regardless of how the published market states are combined with any side information, even including the full list of all other participant's trades.
(This can be formalized using the notion of \emph{joint differential privacy}, but for simplicity we will not do so.)

\paragraph{Tool 1: generalized Laplace noise.}
Imagine that the market could first collect all $T$ trades simultaneously, then sum them and publish some $\hat{q}^T$, a noisy version of the market state $q^T = \sum_{t=1}^T dq^t$.

In this scenario, there is only one output $\hat{q}^T$ instead of a whole list of outputs $\hat{q}^1,\dots$.
The standard, simplest solution to protecting privacy would be to take the true sum $q^T$ and add noise from a generalization of the Laplace distribution.
The real-valued $Lap(b)$ random variable has probability density $x \mapsto \frac{1}{2b} e^{-|x| / b}$.
In $\reals^d$, given a norm $\|\cdot\|$, we define the generalized $Lap_d(b)$ distribution to have probability density proportional to $e^{\|x\| / b}$.
In this case, releasing $\hat{q} = q + Lap_d(1/\epsilon)$ is $\epsilon$-differentially private: Given $q,q'$ with $\|q-q'\| \leq 1$, the ratio of probability densities at any $\hat{q}$ is $e^{\epsilon \|q-\hat{q}\|} / e^{\epsilon \|q' - \hat{q}\|} \leq e^{\epsilon \|q - q'\|} \leq e^{\epsilon}$.
(When the norm is $L^1$, this corresponds to independent scalar $Lap(1/\epsilon)$ noise on each of the $d$ coordinates.)
Note that it also satisfies a good accuracy guarantee, as the amount of noise required does not scale with $T$; so with enough participants, this mechanism becomes a very accurate indication of the ``average'' trade while still preserving privacy.

\paragraph{Tool 2: continual observation technique.}
Unfortunately, the above solution is not sufficient because our market must publish a market state at each time step.
One naive approach is to apply the above solution independently at each time step, \emph{i.e.} produce each $\hat{q}^t = q^t + z^t$ where $z^t$ contains independent Laplace noise.
The problem is that each step reveals more information about a trade, for instance, $dq^1$ participates in $T$ separate publications.
To continue preserving privacy, each $z^t$ must have a much larger variance, which makes the published market states very inaccurate.

A second naive approach is to add noise to each $dq^t$ just once, producing $\hat{dq}^t = dq^t + z^t$.
Then set $\hat{q}^t = \sum_{s=1}^t \hat{dq}^t$.
The benefit to this approach is that it can re-use the noisy $z^t$ variables across time steps, rather than re-drawing new noise each time.
The problem is that, while each $z^t$ is small in magnitude, there are many of them; for example, the final $\hat{q}^T$ contains $T$ pieces of noise, which add up to a very inaccurate estimate of the true market state.
This contrasts with the first naive approach, in which each publication only includes one piece of noise, but that piece of noise is very large.

The idea of the ``continual observation'' technique, pioneered by \citet{dwork2010differential} and \citet{chan2011private}, is to strike a balance between these extremes by re-using noise a limited number of times while also keeping each piece of noise small.
Roughly, each publication $\hat{q}^t$ will include a logarithmic (in $t$) number of pieces of noise, each of which is only ``logarithmically large''.

\begin{definition}
We define the \emph{private market mechanism} for observation $Z \in \Z$ and securities $\phi: Z \to \reals^d$ using a cost function $C$ with parameter $\lambda$.
At each time $t$ participant $t$ arrives and proposes trade $dq^t$ with $\|dq^t\| \leq 1$, unobservable to all others.
At most $T$ participants may arrive.
Let $q^t = \sum_{j=1}^t dq^j$.
Let $z^t = 0$ and for all $t \geq 1$, let $z^t \sim Lap_d(2\lceil \log T\rceil /\epsilon)$.
At each time $t$, the mechanism publishes market state
  \[ \hat{q}^t := q^t + z^t + z^{s(t)} + z^{s(s(t))} + \cdots + z^0 , \]
where $s(t)$ is defined by writing the integer $t$ in binary, then flipping the rightmost ``one'' bit to zero.
Participant $t$ is charged $C(\hat{q}^t + dq^t) - C(\hat(q)^t)$, unobservable to all others.
When outcome $Z$ is observed, she is paid $dq^t \cdot \phi(Z)$, unobservable to all others.
\end{definition}

We note that $s(0) = 0$ and $s(t) < t$ for all $t > 0$.
A convenient notation is to let
  \[ \hat{q}^{s(t):t} := \left(\sum_{j=s(t)+1}^t dq^j\right) ~ + ~ z^t . \]
Then we can define the mechanism recursively as
\begin{align*}
  \hat{q}^t &= q^t + z^t + z^{s(t)} + z^{s(s(t))} + \cdots + z^0  \\
  &= \hat{q}^{s(t):t} + \hat{q}^{s(t)} .
\end{align*}

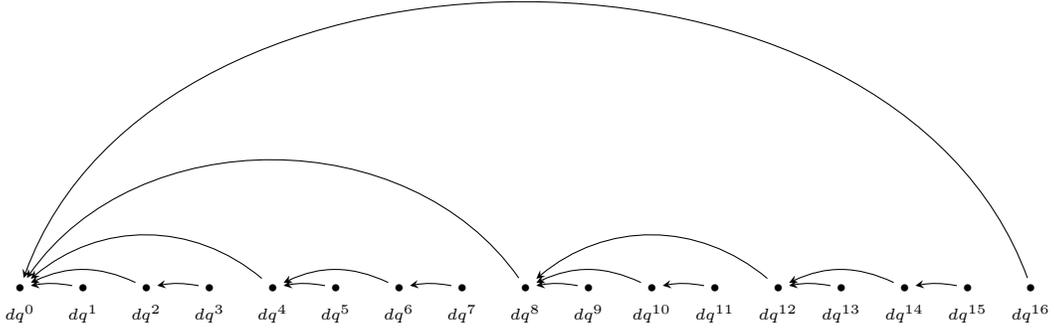
\begin{figure}\tiny
  \centering
  \begin{tikzpicture}[scale=0.7]
    \node (n0) at (0,0) {$\bullet$};
    \node at (0,-0.5) {$dq^0$};
    \foreach \j [count=\i] in {0,0,2,0,4,4,6,0,8,8,10,8,12,12,14,0}
    {
      \node(n\i) at (1.2*\i,0) {$\bullet$};
      \node at (1.2*\i,-0.5) {$dq^{\i}$};
      \pgfmathsetmacro{\angle}{10 + 15 * log2(\i-\j)}
      \draw (n\i) edge[->,>=stealth,in=\angle,out=180-\angle] (n\j);
    }
  \end{tikzpicture}
  \caption{\small  Picturing the continual observation technique for preserving privacy~\citep{dwork2010differential,chan2011private}.
   Each $dq^t$ is a trade.
   The true market state at $t$ is $q^t = \sum_{j=1}^t dq^j$ and the goal is to release a noisy version $\hat{q}^t$
   Each arrow originates at $t$, points backwards to $s(t)$, and is labeled with independent Laplace noise vector $z^t$.
   Now $\hat{q}^t = q^t + z^t + z^{s(t)} + z^{s(s(t))} + \cdots$.
   In other words, the noise added at $t$ is a sum of noises obtained by following the arrows all the way back to $0$.
   There are two key properties: Each $t$ has only $\log T$ arrows passing above it, and each path backwards takes only $\log T$ jumps.}
  \label{fig:continual-tree}
\end{figure}

\paragraph{Remark.}
Notice that $\lambda$ has no impact on the construction of the market, in particular does not affect the amount of noise to add.
Intuitively, this is because the market is defined entirely in ``share space'', while price sensitivity relates shares to prices.
We will not need to discuss $\lambda$ until we discuss accuracy of the prices, which is irrelevant to the proof of privacy.

\begin{theorem}[Privacy] \label{thm:classic-private}
  Assuming that all trades satisfy $\| dq^t \| \leq 1$ (under the same norm as used for the generalized Laplace distribution), the private mechanism is $\epsilon$-differentially private in the trades $dq^1,\dots,dq^T$ with respect to the output $\hat{q}^1,\dots,\hat{q}^T$.
\end{theorem}
\begin{proof}
We emphasize that this proof does not contain new ideas beyond the original continual observation technique, but merely adapts them to this setting.

We will imagine that the market publishes every partial sum it uses, \emph{i.e.} $\hat{q}^{s(t):t}$ for all time steps $t$.
The actual published values $\hat{q}^t$ are functions of these outputs, so by the \emph{post-processing} property of differential privacy~\citep{dwork2014algorithmic}, it suffices to show that publishing each of these partial sums would be $\epsilon$-differentially private.

The idea is to treat each publication of the form $\hat{q}^{s(t):t}$ as a separate mechanism, which has (claim 1) a guarantee of $(\epsilon/\lceil \log T \rceil)$-differential privacy.
We then show (claim 2) that any one trade $dq^t$ participates in at most $\lceil \log T \rceil$ of these mechanisms.
These two claims imply the result because, by the composition property of differential privacy~\citep{dwork2014algorithmic}, each trade is therefore guaranteed $\lceil \log T\rceil \cdot (\epsilon / \lceil \log T \rceil) = \epsilon$-differential privacy.

First, we claim that each publication $\hat{q}^{s(t):t}$ preserves $\epsilon / \lceil \log T \rceil$ differential privacy of each trade $dq^{t'}$ that participates, i.e. with $s(t) < t' \leq t$.
This follows from the definition of $\hat{q}^{s(t):t}$ because, since $\|dq^{t'}\| \leq 1$, an arbitrary change in $dq^{t'}$ changes norm of the partial sum of trades by at most $2$, and $z^t$ is a draw from the generalized $Lap_d(2\lceil \log T \rceil/\epsilon)$ distribution with respect to the same norm.

Second, we claim that each trade $dq^{t'}$ participates in at most $\lceil \log T \rceil$ different partial sums $\hat{q}^{s(t):t}$.
To show this, we only need to count the time steps $t$ where $s(t) < t' \leq t$, in other words, integers $t \geq t'$ where zeroing the rightmost ``one'' bit gives a number less than $t'$.

Without loss of generality, the the binary expansion of $t$ is $b_m b_{m-1} \ldots b_j 1 0 \ldots 0$ for some $m,j$ and then $s(t)$ has expansion $b_m b_{m-1} \ldots b_j 0 0 \ldots 0$.
Hence the condition $s(t) < t' \leq t$ implies that the binary expansion of $t$ matches that of $t'$ from bits $m$ to $j$, then has a one at bit $j-1$, and has zeroes at all lower-order bits.
Since $m$ is fixed for $t'$, this can only happen once for each $j$, or at most $m$ total times; and $m \leq \lceil \log T \rceil$ because $t' \leq T$.
\end{proof}

\begin{lemma}[Accuracy of share vector] \label{lem:classic-private-sharevec-accurate}
  In the private mechanism with $L^1$ norm, $d$ securities, and $T$ time steps, we have with probability $1-\gamma$,
  \[ \max_{t} \|q^t - \hat{q}^t\|_1 \leq \frac{4\sqrt{2} d \log \lceil T \rceil}{\epsilon} \ln\left(\frac{2 T d}{\gamma}\right) . \]
\end{lemma}
\begin{proof}
  For each $t$, each coordinate $i$ of $q^t - \hat{q}^t$ is the sum of at most $\lceil \log T \rceil$ independent variables distributed $Lap(2 \lceil \log T \rceil / \epsilon)$.
  We will choose $\beta$ such that each coordinate's absolute value exceeds $\beta$ with probability $\frac{\gamma}{Td}$; there are $d$ coordinates per time step and $T$ time steps, so a union bound gives the result.

  Choose $\beta$ such that, if $Y$ is the sum of $k = \lceil \log T \rceil$ independent $Lap(b)$ variables with $b = 2 \lceil \log T \rceil / \epsilon)$ variables, then
    \[ \Pr[|Y| > \beta] = \gamma' . \]
  A concentration bound for the sum of $k$ independent $Lap(b)$ variables, Corollary 12.3 of \citet{dwork2014algorithmic}\footnote{In the parameters of that Corollary, we choose $\nu = b \sqrt{\ln(2/\gamma')}$ as we will have $\ln(2/\gamma') > k$.} gives
    \[ \beta \leq 2\sqrt{2} b \ln \frac{2}{\gamma'} . \]
  Now choose $\gamma' = \frac{\gamma}{Td}$.
  To recap, each $|q^t(i) - \hat{q}^t(i)| \leq \beta$ except with probability $\gamma' = \frac{\gamma}{Td}$, hence by a union bound this holds for all $t,i$ except with probability $\gamma$, hence $\|q^t - \hat{q}^t\|_1 \leq d \beta$ except with probability $\gamma$.
\end{proof}

As mentioned above, the previous results (Theorem \ref{thm:classic-private} and Lemma \ref{lem:classic-private-sharevec-accurate}) do not depend on $\lambda$ at all, because they do not mention the prices.
We now ask what a ``reasonable'' choice of $\lambda$ can be so that the prices are interpretable as predictions, i.e. the prices are ``accurate''.
\begin{theorem}[Accuracy of prices] \label{thm:classic-private-accurate}
  In the private mechanism, let $p^t = \nabla C(q^t)$ and let $\hat{p}^t = \nabla C(\hat{q}^t)$.
  Then to satisfy $\|p^t - \hat{p}^t\|_1 \leq \alpha$ for all $t$, except with probability $\gamma$, it suffices for the price sensitivity to be
   \[ \lambda^* = \frac{\alpha ~ \epsilon}{4\sqrt{2} d \lceil \log T \rceil \ln(2Td/\gamma)} ~. \]
\end{theorem}
\begin{proof}
  By definition of $\lambda$, we have
  \begin{align}
    \|p^t - \hat{p}^t\|_1 &\leq \lambda \|q^t - \hat{q}^t\|_1  \nonumber \\
                          &\leq \lambda \frac{4\sqrt{2} d \lceil \log T \rceil}{\epsilon} \ln\left(\frac{2 T d}{\gamma}\right)  \label{eqn:lambda-for-accuracy}
  \end{align}
  for all $t$ except with probability $\gamma$, by Lemma \ref{lem:classic-private-sharevec-accurate}.
  We now just choose $\lambda$ so that $(\ref{eqn:lambda-for-accuracy}) \leq \alpha$.
\end{proof}

\section{Slowly-Growing Budget}

We prove the incentive and budget claims in separate lemmas.
\begin{lemma}[Incentive to trade] \label{lemma:fee-incentive}
  In the private market with transaction fee $\alpha$, a participant at time $t$ having belief $p$ with $\|p - \hat{p}^t\|_{\infty} \geq 2\alpha$ can make a strictly positive expected profit by participating.
\end{lemma}
\begin{proof}
  Ignoring the transaction fee, the expected profit from purchasing $dq$ (recall $\|dq\|_1 \leq 1$) is 
  \begin{align*}
    \text{profit} ~ &= \langle dq, p \rangle + C(\hat{q}^t) - C(\hat{q}^t + dq) .
  \end{align*}
  Because $C$ is convex, $C(\hat{q}^t) - C(\hat{q}^t + dq) \geq \langle \nabla C(\hat{q}^t + dq),-dq\rangle$.
  So
  \begin{align*}
    \text{profit} ~ &\geq \langle dq, p - \nabla C(\hat{q}^t + dq) \rangle  \\
                    &=    \langle dq, p - \hat{p}^t \rangle - \langle dq, \nabla C(\hat{q}^t + dq) - \hat{p}^t \rangle .
  \end{align*}
  By H\"{o}lder's inequality and the definition of price sensitivity,
  \begin{align*}
    \langle dq, \nabla C(\hat{q}^t + dq) - \hat{p}^t \rangle
      &\leq \|dq\|_{\infty} \| \nabla C(\hat{q}+dq) - \hat{p}\|_1  \\
      &\leq \|dq\|_{\infty} \lambda \|dq\|_1  \\
      &\leq \lambda .
  \end{align*}
  So we have
  \begin{align*}
    \text{max profit} ~ &\geq \max_{dq: \|dq\|_1 \leq 1} \langle dq, p - \hat{p}^t \rangle - \lambda  \\
                        &=    \|p - \hat{p}\|_{\infty} - \lambda  \\
                        &\geq 2\alpha - \lambda  \\
                        &>    \alpha
  \end{align*}
  as $\lambda < \alpha$ by construction.
  Because there exists a trade with expected profit strictly above $\alpha$, the trader has an incentive to pay the $\alpha$ transaction fee and participate.
\end{proof}

\begin{claim}[Claim~\ref{claim:traders-between-noise-bundle}] \label{claim:traders-between-noise-bundle-app}
    For each $t$, exactly $b(t)$ traders arrive between the purchase and the sale of bundle $z^t$; furthermore, $q_{\text{sell}}^t - q_{\text{buy}}^t$ is exactly equal to the sum of these participants' trades.
  \end{claim}
  \begin{proof}
    Note that if we write $t$ in binary, it has a one in the bit position $\log b(t)$, followed by zeros.
    By definition of the algorithm, $z^t$ is sold at the next time $t' > t$ where the bit $\log b(t)$ is flipped to zero.
    So we have $t' - t = b(t)$, so $b(t)$ traders have arrived.

    Now we want to show that $q_{\text{sell}}^t - q_{\text{buy}}^t$ is the sum of their trades, i.e. that every noise trade bundle $z^{t'}$ held by the trader before buying $z^t$ is still held at the moment of selling $z^t$, and no other noise trade bundles are held at that time.

    Consider all of the noise bundles that were already held at time $t$ (after selling the appropriate bundles at that time, but before purchasing $z^t$).
    By definition of the algorithm, these were purchased at times $s$ with $b(s) > b(t)$, so by the above discussion, they are not sold until times $t'$ where the bits in positions $\log b(s)$ are flipped to zero, which cannot happen until after bit $\log b(t)$ is flipped and $z^t$ is sold.
    Meanwhile, every bundle purchased after $z^t$ is sold by, at the latest, the same time that $z^t$ is sold, as they correspond to lower-order bits; and any sold at the same time as $z^t$ are sold first because they were purchased later.
  \end{proof}

  \begin{claim}[Claim~\ref{claim:max-gap}] \label{claim:max-gap-app}
    If the noise trader purchases and later sells $z^t$, then her net loss in expectation over $z^t$ (but for any trader behavior in response to $z^t$), is at most $\lambda b(t) K$ where $K = \E \|z^t\|_2$.
  \end{claim}
  \begin{proof}
    Given the noise trader's bundle drawn is $z^t$, her loss is:
      \[ C(q_{\text{buy}}^t + z^t) - C(q_{\text{buy}}^t) + C(q_{\text{sell}}^t) - C(q_{\text{sell}}^t + z^t) . \]
    The first pair of terms represents the payment made to purchase $z^t$ (moving the market state from $q_{\text{buy}}^t$); the second pair represents the payment to sell $z^t$ (moving the state to $q_{\text{sell}}^t$).
    Claim \ref{claim:traders-between-noise-bundle-app} implies that $\|q_{\text{buy}}^t - q_{\text{sell}}^t\|_1 \leq b(t)$, as each trader can buy or sell at most $1$ unit of shares.
    Therefore, the net loss on bundle $z^t$ is at most
    \begin{align*}
      &\E_{z^t}   ~ \max_{q, q' : \|q - q'\|_1 \leq b(t)} C(q + z^t) - C(q) + C(q') - C(q' + z^t)
    \end{align*}
    Now, we have
    \begin{align*}
      C(q + z^t) - C(q)
        &= \int_{x=0}^1 \nabla C(q + x z^t) \cdot z^t dx , \\
      C(q + r + z^t) - C(q + r)
        &= \int_{x=0}^1 \nabla C(q + r + x z^t) \cdot z^t dx .
    \end{align*}
    So the difference is
    \begin{align*}
      &\int_{x=0}^1 \left\langle \nabla C(q + x z^t) - \nabla C(q + r + x z^t) ~,~ z^t \right\rangle dx \\
      &\leq \int_{x=0}^1 \lambda \|r\|_2 \|z^t\|_2 dx  \\
      &= \lambda \|r\|_2 \|z^t\|_2
    \end{align*}
    by definition of price sensitivity $\lambda$.
    We also have $\|r\|_2 \leq \|r\|_1 \leq b(t)$.
    This bound holds for each outcome of $z^t$ and any behavior of the participants, so we conclude the lemma statement, that expected loss is bounded by $\lambda b(t) \E \|z^t\|_2$.
  \end{proof}

\section{Constant Budget Bound}

\begin{lemma} \label{lemma:tlogt-bound}
  Let $A,D$ be constants at least $1$ with $AD \geq 5$.
  Then for all $T \geq 9A \left(\ln(AD)\right)^2$, we have $T \geq A \left(\ln(TD)\right)^2$.
\end{lemma}
\begin{proof}
  Let $T^* = 9A\left(\ln(AD)\right)^2$.
  First, we prove the inequality for $T^*$:
  \begin{align*}
    T^* &= 9A\left(\ln(AD)\right)^2  \\
        &= A \left(3\ln(AD)\right)^2  \\
        &= A \left(\ln(AD (AD)^2) \right)^2.
  \end{align*}
  Now for all $AD \geq 5$, we have $AD \geq 3 \ln(AD)$, so
  \begin{align*}
    T^* &\geq A \left(\ln\left(9AD (\ln(AD))^2\right)\right)^2  \\
        &=    A \left(\ln(T^*D)\right)^2 ,
  \end{align*}
  as desired.
  Now we wish to extend this to all $T \geq T^*$.
  Compare the derivative of the left side, $\frac{dT}{dT} = 1$, with that of the right side:
  \begin{align*}
    \frac{d}{dT}\left(A \left(\ln(TD)\right)^2\right)
      &=    \frac{2A \ln(TD)}{T}  \\
      &\leq \frac{2}{\ln(TD)} \leq 1
  \end{align*}
  at $T = T^*$.
  Now if the inequality holds for all $T' \in [T^*,T)$, then it holds for $T$ as the left side only increases more quickly than the right.
  So by transfinite induction, it holds for all $T \geq T^*$.
\end{proof}

\begin{proof}[Proof of Theorem \ref{thm:main}]
Fix the parameters $\epsilon$ and $d$ throughout.
Let $\lambda^*(T, \alpha, \gamma)$ be the price sensitivity parameter as a function of these variables given in Theorem \ref{thm:classic-private-accurate}.

$\epsilon$-differential privacy of the market follows by the post-processing property of differential privacy~\citep{dwork2014algorithmic} because each stage $k$ is differentially private for the participants who arrive in that stage.
All information released after stage $k$ depends on these participants only through the noisy market state at the end of stage $k$, which is $\epsilon$-d.p.

To show the incentive guarantee, note that the transaction fee is always fixed at $\alpha$, so the incentive proof of Theorem \ref{thm:private-loss-one-over-lambda} goes through immediately.

To show the accuracy guarantee, note the the prices up to $T^{(1)}$ arrivals satisfy an $\alpha/2$ guarantee; therefore the starting prices of the new market are within $\alpha/2$ of what they would be without added noise.
The  prices up to $T^{(2)}$ additional arrivals are within $\alpha/2 + \alpha/4$ of what they would have been (since they begin within $\alpha/2$ and are designed to stay within $\alpha/4$ of this shifted goal); and so on, telescoping to at most $\alpha$.
Similarly, the chance of failure of any of these guarantees, by a union bound, is at most $\gamma/2 + \gamma/4 + \cdots \leq \gamma$.

Now we must show bounded worst-case loss, and how to set $T^{(k)}$.
We will choose $T^{(1)}$ to be a constant and each $T^{(k)} = 4T^{(k-1)}$.

We will claim two things:
\begin{enumerate}
  \item In the final stage $k$ where not all participants arrive, the market maker's loss is at most $\frac{\alpha}{16}T^{(k)}$.
  \item In each stage $k$ that is completed (all $T^{(k)}$ participants arrive), the market maker's profit from that stage is at least $\frac{\alpha}{2} T^{(k)}$.
\end{enumerate}
These together prove bounded worst-case loss: If at least one stage is completed, the total profit is in fact positive: it is positive from all but the last stage, whose loss is at most $\frac{\alpha}{16}T^{(k)} \leq \frac{\alpha}{4} T^{(k-1)}$ which is smaller than the profit made in stage $k-1$.
If no stages are completed, i.e. fewer than $T^{(1)}$ participants arrive, then expected worst-case loss is bounded by $\frac{B'}{\lambda^{(1)}}$.
This gives a budget bound of $\frac{B'}{\lambda^{(1)}}$, which will be computed below.

\paragraph{Proof of (1).} First, we must prove that the worst-case loss in stage $1$ is at most $\frac{\alpha}{16}T^{(1)}$.
In doing so, we will explicitly compute a sufficient $T^{(1)}$ and this worst-case loss.
Then, we must show the same fact for all other stages.

The worst-case loss in stage $1$, by Theorem \ref{thm:private-loss-one-over-lambda}, is
\begin{align*}
  B &= \frac{B'}{\lambda^{(1)}}  \\
    &= B' \frac{8\sqrt{2} d \lceil \log T^{(1)} \rceil \ln\left(4 T^{(1)} d / \gamma\right)}{\alpha ~ \epsilon}  \\
    &\leq B' \frac{8\sqrt{2} d \left(\ln\left(4 T^{(1)} d / \gamma\right)\right)^2}{\alpha ~ \epsilon} .
\end{align*}
For convenience, set $A' = B' \frac{8 \sqrt{2} d}{\alpha \epsilon}$ and set $D = 4d / \gamma$.
Then we have $B \leq A' \left(\ln(T^{(1)}D)\right)^2$; to prove claim (1) for stage $1$, we wish to pick $T^{(1)}$ such that $B \leq \frac{\alpha}{16}T^{(1)}$.
Setting $A = \frac{16}{\alpha}A'$, we need to have $A \left(\ln(T^{(1)}D)\right)^2 \leq T^{(1)}$.
By Lemma \ref{lemma:tlogt-bound}, this holds for
\begin{align*}
  T^{(1)} &= 9A \left(\ln(AD)\right)^2  \\
          &= B' \frac{1152\sqrt{2} d \left(\ln \frac{4608 B' \sqrt{2} d^2}{\gamma \alpha^2 \epsilon} \right)^2}{\alpha^2 ~ \epsilon}
\end{align*}
So for this choice of $T^{(1)}$, we have claim (1) for stage $1$.
We note the budget bound is $B \leq \frac{\alpha}{16}T^{(1)}$.
Now we just show that $T^{(k)}$ increases faster than $1/\lambda^{(k)}$.
$T^{(k)} = 4T^{(k-1)}$, but
\begin{align*}
  \frac{B}{\lambda^{(k)}} &= \frac{4\sqrt{2} B 2^k d \lceil \log T^{(k)} \rceil \ln\left(2 T^{(k)} d 2^k / \gamma\right)}{\alpha ~ \epsilon} \\
    &= 2 \frac{4\sqrt{2} B 2^{k-1} d \lceil 2 + \log T^{(k-1)} \rceil \left(\ln(8) + \ln\left(2 T^{(k-1)} d 2^{k-1} / \gamma\right)\right)}{\alpha ~ \epsilon}  \\
    &\leq 4 \frac{1}{\lambda^{(k-1)}}
\end{align*}
for sufficiently large $T^{(k-1)}$, i.e. if $T^{(1)}$ is a sufficiently large constant.
So $\frac{B}{\lambda^{(k)}}$ grows more slowly than $T^{(k)}$ and the inequality $\frac{B}{\lambda^{(k)}} \leq \frac{\alpha}{16} T^{(k)}$ continues to hold.

\paragraph{Proof of (2).}
Let us lower-bound the profit in stage $k$ if completed.
By Inequality \ref{eqn:exact-wcl-fee-bound} (from Theorem \ref{thm:private-loss-one-over-lambda}), that the market-maker profit if $T' = T^{(k)}$ participants arrive is
\begin{align*}
  &T^{(k)} \left(c - K \lambda^{(k)} \log T^{(k)}\right) - \frac{B}{\lambda^{(k)}}  \\
  &= T^{(k)} \left(\alpha - \frac{\sqrt{2d} \lceil \log T^{(k)} \rceil \log T^{(k)}}{\epsilon} \frac{(\alpha/2^k) ~ \epsilon}{4\sqrt{2} d \lceil \log T^{(k)} \rceil \ln(2d T^{(k)} 2^k/ \gamma)} \right) - \frac{B}{\lambda^{(k)}}  \\
  &= \alpha T^{(k)} \left(1 - \frac{\log T^{(k)}}{4 (2^k) \sqrt{d} \ln(2 d T^{(k)} 2^k / \gamma)} \right) - \frac{B}{\lambda^{(k)}} .
\end{align*}
Recall that $\frac{B}{\lambda^{(k)}} \leq \frac{\alpha}{16}T^{(k)}$.
We want to conclude that the profit in stage $k$ is at least $\frac{\alpha}{2}T^{(k)}$, so we just need to show that
  \[ 1 ~ - ~ \frac{\log T^{(k)}}{4 (2^k) \sqrt{d} \ln(2 d T^{(k)} 2^k / \gamma)} ~ - \frac{1}{16} \geq \frac{1}{2} . \]
The fraction is decreasing in $k$, so it suffices to achieve this for $k=1$, $d=1$, and $\gamma=1$, where we have
  \[ \frac{\log T^{(1)}}{8 \ln(4 T^{(1)})} \leq \frac{1}{4} . \]
This suffices to prove Claim (2).
\end{proof}

\end{document}